%% file: conference_101719.tex
\documentclass[conference]{IEEEtran}
\IEEEoverridecommandlockouts

\usepackage{cite}
\usepackage{amsmath,amssymb,amsfonts,amsthm,mathtools}
\usepackage{algorithmic}
\usepackage{graphicx}
\usepackage{textcomp}
\usepackage{xcolor}
\usepackage{amsmath,amsthm,enumitem,nicefrac,bbm,verbatim,amssymb,amsfonts,amscd, graphicx,hyperref,float,mathtools,bm,url,tikz,stackrel}
\input{glodef}
\input{locdef}
\def\BibTeX{{\rm B\kern-.05em{\sc i\kern-.025em b}\kern-.08em
    T\kern-.1667em\lower.7ex\hbox{E}\kern-.125emX}}

\newcommand*{\EXTENDED}{}%

\begin{document}

\title{On One-Bit Quantization\\
{\footnotesize 
\thanks{This
research was supported by the
US National Science Foundation under grants
 CCF-2008266 and CCF-1934985, by the US Army
 Research Office under grant W911NF-18-1-0426,
and by a gift from Google.
}}}


\author{\IEEEauthorblockN{Sourbh Bhadane and Aaron B. Wagner} 
	\IEEEauthorblockA{\textit{School of Electrical and Computer Engineering} \\
		\textit{Cornell University}\\
		Ithaca, NY 14853 USA\\
	Email: \{snb62,wagner\}@cornell.edu}
%
}

\maketitle

\begin{abstract}
    We consider the one-bit quantizer that minimizes the mean squared error for a source living in a real Hilbert space. The optimal  quantizer is a projection followed by a thresholding operation, and we provide methods for identifying the optimal direction along
    which to project. As an application of our methods, we characterize the optimal one-bit quantizer for a continuous-time
    random process that exhibits low-dimensional structure. We numerically show that this optimal quantizer is found by
    a neural-network-based compressor trained via stochastic gradient descent.
\end{abstract}


\section{Introduction}\label{sec:intro}
\input{intro}

\section{Preliminaries}\label{sec:prelim}
\input{preliminaries}

\section{General Methods}\label{sec:amen} 
\input{amenability}

\section{The Stationary Sawbridge}\label{sec:stsaw}
\input{stsaw}
\section{Numerical Results}\label{sec:nn}
\input{expt}
\section*{Acknowledgment}
The second author wishes to thank Johannes Ball\'{e}
for helpful discussions.

\bibliographystyle{IEEEtran}
\bibliography{IEEEabrv,biblio}

\end{document}

%% file: glodef.tex
\newtheorem{Theorem*}{Theorem}

\newtheorem{Claim*}[Theorem]{Claim}

\newtheorem{CounterExample*}{$\overline{\hbox{\bf Example}}$}

\newtheorem{Example*}[Theorem]{Example}

\newtheorem{Intuition*}[Theorem]{Intuition}
\newtheorem{Joke*}[Theorem]{Joke}

\newtheorem{Lemma*}[Theorem]{Lemma}
\newtheorem{Open problem}[Theorem]{Open problem}

\newtheorem{Question*}[Theorem]{Question}

\usepackage{fullpage}

\newtheorem{theorem}{Theorem}

\newtheorem{lemma}[theorem]{Lemma}

\newtheorem{definition}[theorem]{Definition}

\def \bSubexa    {\begin{subexa}}

\usepackage[colorinlistoftodos,textsize=scriptsize]{todonotes}

\newcommand{\ignore}[1]{}

\newcommand{\EE}{\mathbb{E}}

\newcommand{\RR}{\mathbb{R}}

\def \cH     {{\cal H}}

\newcommand{\Var}{{\rm Var}}

\def\ignore#1{}

\newcommand{\bi}{\begin{itemize}}
\newcommand{\ei}{\end{itemize}}

\def\orpro{\mathop{\mathchoice
   {\vee\kern-.49em\raise.7ex\hbox{$\cdot$}\kern.4em}
   {\vee\kern-.45em\raise.63ex\hbox{$\cdot$}\kern.2em}
   {\vee\kern-.4em\raise.3ex\hbox{$\cdot$}\kern.1em}
   {\vee\kern-.35em\raise2.2ex\hbox{$\cdot$}\kern.1em}}\limits}

\def\andpro{\mathop{\mathchoice
 {\wedge\kern-.46em\lower.69ex\hbox{$\cdot$}\kern.3em}
 {\wedge\kern-.46em\lower.58ex\hbox{$\cdot$}\kern.25em}
 {\wedge\kern-.38em\lower.5ex\hbox{$\cdot$}\kern.1em}
 {\wedge\kern-.3em\lower.5ex\hbox{$\cdot$}\kern.1em}}\limits}

\def\simge{\mathrel{%
   \rlap{\raise 0.511ex \hbox{$>$}}{\lower 0.511ex \hbox{$\sim$}}}}

\def\simle{\mathrel{
   \rlap{\raise 0.511ex \hbox{$<$}}{\lower 0.511ex \hbox{$\sim$}}}}

\providecommand{\email}[1]{\href{mailto:#1}{\nolinkurl{#1}\xspace}}

%% file: locdef.tex
\newcommand{\eps}{\varepsilon}

\DeclareMathOperator*{\argmax}{arg\,max}

\newcommand{\drop}{U}
\newcommand{\phase}{V}
\newcommand{\stsaw}{Y}
\newcommand{\saw}{X}
\newcommand{\rphase}{R}
\newcommand{\indep}{\perp \!\!\! \perp}

\newcommand{\eqdef}{\stackrel{\text{def}}{=}}

\newcommand{\thresh}{T}
\newcommand{\proj}{Z}
\newcommand{\projDC}{Z_{DC}}
\newcommand{\projAC}{Z_{AC}}
\newcommand{\lebtwo}{L^2\left[0,1\right]}

\newcommand{\acorr}{K}
\newcommand{\coeff}{G}
\newcommand{\gcoeff}{c}
\newcommand{\eigfunc}{\psi}
\newcommand{\eigval}{\lambda}
\newcommand{\var}[1]{\text{Var}\left(#1\right)}

\newcommand{\amen}{\zeta}

\newcommand{\vdrop}{\text{Vardrop}}
\newcommand{\hil}{\cH}
\newcommand{\qcell}{A}
\newcommand{\rec}{\hat{x}}
\newcommand{\modproj}{\widetilde{\proj}}

\newcommand{\vthresh}{w}
\newcommand{\sawDC}{U_{DC}}

%% file: intro.tex
The classical theory of lossy compression is based on the analysis
of stationary Gaussian sources with a mean-squared error distortion measure.
Standard results stipulate that a near-optimal method for compressing
such sources is to apply a linear whitening transform, followed
by a uniform quantizer, follwed by entropy coding~\cite[Sec.~5.5]{Pearlman:Said}. 
This is indeed the approached adopted by various practical compression
standards.

Recently, lossy compression methods based on Artificial Neural
Networks have begun to outperform those that use the classical
approach for images~(e.g.,~\cite{BalleCMSJAHT20}) and other 
sources~. Since the classical approach is provably near-optimal for Gaussian
sources, ANN-based methods are evidently able to exploit
non-Gaussianity in practical sources of interest. This calls
for a shift away from Gaussian sources toward ones that can
better explain the performance of ANN-based codes (e.g.~\cite{WagnerB21}).

Analyzing such models can be challenging (with~\cite{WagnerB21} being
a notable exception), leading one to focus on high-rate
and low-rate regimes. In this paper, we focus on the latter, 
specifically the characterization of the optimal one-bit 
quantizer for a given source under mean squared error (MSE).

Despite the simplicitly with which this problem can be stated,
relatively little is known about it. For log-concave densities,
there exists a unique locally optimal quantizer, which can
be found using the Lloyd-Max 
algorithm~\cite{Fleischer64,Trushkin82,Kieffer83,Max60,Lloyd82}. 
For sources with a density of the form $f(x) = g(x^T K x)$, where 
$g(\cdot)$ is decreasing and $K$ is positive semidefinite,
Magnani \emph{et al.}~\cite{MagnaniGG05} show that the optimal 
reconstructions lie on the major axis of the ellipsoid 
associated with $K$. On the other hand, it is known that
the optimal quantizer is not necessarily symmetric about
0 even if the distribution itself is. Consider the
distribution that is uniformly distributed across
the three points $\{-1,0,1\}$. It is elementary
to check that the best symmetric quantizer is outperformed
by one that maps to the closest 
reconstruction among the set $\{-1,1/2\}$. See
Abaya and Wise~\cite{Abaya:Counter} for an earlier example
that is continuous and monotonically decreasing
(cf.~\cite{MagnaniGG05}).

We develop results toward a general theory of
optimal one-bit quantization. Any optimal one-bit
quantizer can evidently be implemented via a projection
operation followed by a thresholding. We follow
Magnani \emph{et al.} in the sense that we focus on
identifying the best direction in which to project;
once this is identified, the optimal threshold can
be found by a one-dimensional sweep. The optimal
direction is controlled by a tension between the
variance of the projected source and its ``amenability'' 
to one-bit quantization. On the one hand, 
quantizing high-variance directions results
in a larger variance drop, i.e., a lower MSE.
On the other hand, for a given variance,
some distributions result in a lower variance
drop under one-bit quantization than others
(consider, for example, a standard Normal
versus the uniform distribution on 
$\{-1,1\}$; see~\cite{WagnerB21} for a naturally-occurring
example). We provide methods for resolving
this tension, which we demonstrate on an example
random process called the \emph{stationary sawbridge}.
For this infinite-dimensional process we characterize
the optimal one-bit quantizer. Moreover, we show
that it is found by an off-the-shelf ANN compressor
trained via stochastic gradient descent (SGD).

Most of the prior work on optimal one-bit quantizers
focuses on communication instead of compression~\cite{WangFLCL18,MoH15}.
There the objective is to maximize mutual information
or bit error rate instead of MSE. Nonetheless, the 
methods in this paper may have some utility in that
application.

%% file: preliminaries.tex
Let $\hil$ be a real Hilbert space with a countable basis and let $X$ be a random variable in $\hil$. Without loss of generality, we assume throughout that $\EE\left[X\right] = 0$ by which we mean $\EE\left[ \left \langle q,X\right \rangle \right] = 0$ for all $q$ such that $\left \lvert q \right \rvert = 1$, and $\EE\left[ \left \lVert X\right \rVert^2 \right] < \infty$.	
	\begin{definition}
		 A one-bit quantizer is an encoder $f:\hil \mapsto \{ 0,1\}$ and a decoder $g:\{0,1\} \mapsto \hil$. We denote the quantization cells by 
		 \[ \qcell_j = f^{-1}(j)  \hspace{5mm} j \in \{0,1\} , \]
and the reconstructions by 
\[  \rec_j = g(j) \hspace{5mm} j \in \{0,1\}.  \]
We will use $Q$ to 	refer to both $(f,g)$ and $g \circ f$.  $Q$ is said to be a symmetric one-bit quantizer if $\hat{x}_0 = -\hat{x}_1$.
	\end{definition}

We focus on mean-squared error (MSE) as a performance metric. We define the difference between the variance and the infimum of the mean-squared error over all one-bit quantizers as the \emph{variance drop} of a source. 
\begin{definition}
\[	\vdrop_X \eqdef \EE\left[ \left \lVert X \right \rVert^2 \right] - \inf_Q \EE\left[ \left \lVert X - Q(X) \right \rVert^2 \right]. \]
\end{definition}

We will require the notions of symmetric real-valued random variables and log-concave probability density functions (pdf) in the rest of the paper. 

\begin{definition}
	A real-valued random variable $X$ is symmetric if $X$ and $-X$ have the same distribution.
\end{definition}

\begin{definition}
	A probability density function $f:\RR \mapsto \RR_+$ is log-concave if there exists a concave function $\phi:\RR \mapsto \left[-\infty,\infty\right)$ such that for all $x \in \RR$, $f(x) = e^{\phi(x)}$.
\end{definition}

%% file: amenability.tex
The decision boundary of an optimal one-bit quantizer of a random vector is a hyperplane that is normal to the line joining the two reconstructions. Thus one-bit quantization of a random vector can be reduced to projecting the random vector along a direction and thresholding the projection. It would seem natural to project along the direction with the highest variance. 
Yet, as noted in the introduction, lower variance directions might be preferred if they are more amenable
to one-bit quantization. We begin by making this tension precise.
%


\begin{definition}
    The \emph{amenability (to one-bit quantization)} of a real-valued, zero-mean random variable $X$ is defined as $\amen_X \eqdef \frac{\EE\left[ \left \lvert X \right \rvert \right]^2}{\EE\left[X^2\right]}$.
\end{definition}

We note two formal properties of amenability before connecting the
concept to quantization:
\begin{enumerate}
	\item \textbf{Scale-free}: $\amen_X = \amen_{aX}$ for nonzero $a \in \RR$. 
	\item \textbf{Bounded}: $0 \leq \amen_X \leq 1$ where the right hand side inequality follows from Cauchy Schwarz. Both extremes are approachable by distributions with uniformly bounded support. For $X \sim \text{Unif} \{ -1,1\}$, $\amen_X =1$. For the lower limit, consider $X_{\eps,\delta}$ with probability mass function 
	\[ p_{X_{\eps,\delta}}(\pm 1) = \delta, p_{X_{\eps,\delta}}(\pm \eps) = 0.5 - \delta.  \]
	It can be verified that 
	\[ \amen_{X_{\eps,\delta}} = \frac{\EE\left[ \left \lvert X_{\eps,\delta} \right \rvert \right]^2}{\Var_{X_{\eps,\delta}}} = \frac{\left((1-2\delta)\eps + 2\delta\right)^2}{(1-2\delta)\eps^2 + 2\delta}. \]
	Finally, $\lim\limits_{\delta \mapsto 0}\left( \lim\limits_{\eps \mapsto 0} \amen_{X_{\eps,\delta}} \right) = 0$.
\end{enumerate}

The amenability of a few standard distributions whose mean is $0$ is given in Table~\ref{table:amen}. 

\begin{table}[htbp]
	\caption{Amenability of some standard distributions.}
	\begin{center}
		\begin{tabular}{|c|c|}
			\hline
			\textbf{Distribution}& \textbf{Amenability} \\
			\hline
			Unif & $3/4$ \\
			\hline
			Unif*Unif& $2/3$\\ 
			\hline
			Gaussian& $2/ \pi$  \\
			\hline
			Laplacian& $1/2$\\
			\hline
		\end{tabular}
		\label{table:amen}
	\end{center}
\end{table}

%



 A key stepping stone for our main theorem is the relation between  the variance drop of any zero-mean random variable whose optimal one-bit quantizer is symmetric, to its amenability. Note that this relation holds in particular for a symmetric random variable whose pdf is log-concave, since its optimal one-bit quantizer is known to be symmetric~\cite{Kieffer83}.

\begin{lemma}\label{lem:vardrop}
	Let $W$ be a zero-mean, real-valued random variable with a density. Then 
	\begin{enumerate}
		\item 
		\begin{equation} \label{eq:vardrop_nonopti}
 \vdrop_W = \sup\limits_{\vthresh} \left(\EE\left[W \mid W > \vthresh \right]\right)^2 \frac{\Pr\left(W\geq \vthresh\right)}{\Pr\left(W< \vthresh \right)}.
		\end{equation}
		\item Further if $\Pr\left(W \geq 0\right) = \Pr\left(W<0\right) = \frac{1}{2}$ and if an optimal one-bit quantizer is symmetric then 
		\[ \vdrop_W = \EE\left[ \left \lvert W \right \rvert \right]^2 = \amen_W \Var_W. \]
	\end{enumerate}
%
\end{lemma}
\begin{proof}
	Let the quantization cells be $\left( -\infty, \vthresh \right)$ and $\left[\vthresh, \infty\right)$ where $\vthresh \in \RR$. By Lloyd's conditions for local optimality the reconstructions are $\EE\left[ W \mid W < \vthresh \right]$ and $\EE\left[ W \mid W \geq \vthresh \right]$. Therefore the mean-squared error is 
	\begin{align} 
	&\Pr\left( W \geq \vthresh \right) \EE\left[ \left(W - \EE\left[ W \mid W \geq \vthresh \right] \right)^2 \mid W \geq \vthresh \right] \nonumber\\
	&+ \Pr\left(W < \vthresh \right)\EE\left[ \left( W - \EE\left[ W \mid W < \vthresh \right] \right)^2 \mid W < \vthresh \right] \nonumber \\
	&=\EE\left[ W^2 \right] - \Pr\left(W \geq \vthresh \right) \left(\EE\left[ W \mid W \geq \vthresh \right]\right)^2 \nonumber \\
	& \hspace{10mm} -  \Pr\left(W < \vthresh \right) \left(\EE\left[ W \mid W < \vthresh \right]\right)^2 \label{eq:mse_nonsym}
	\end{align}
Since $W$ is zero-mean,
 \begin{align*}
 	&\Pr\left(W\geq \vthresh \right) \EE\left[W \mid W \geq \vthresh \right] \\
 	& \hspace{5mm}+ \Pr\left(W < \vthresh \right) \EE\left[W \mid W < \vthresh \right] = 0.
 \end{align*}
\[\EE\left[W \mid W < \vthresh \right] = -\frac{\Pr\left(W\geq \vthresh \right) \EE\left[W \mid W \geq \vthresh \right]}{\Pr\left(W < \vthresh  \right)}.\] 
Substituting this in \eqref{eq:mse_nonsym} and simplifying, we get 
\[ \vdrop_W = \sup\limits_{\vthresh } \left(\EE\left[W \mid W \geq \vthresh  \right] \right)^2 \frac{\Pr\left(W \geq \vthresh \right)}{\Pr\left(W < \vthresh  \right)}. \]
%
%

When an optimal quantizer is symmetric, we can choose $\vthresh =0$. Therefore, 
\[ \vdrop_W =\EE\left[W \mid W \geq 0 \right]^2 =  \EE\left[ \left \lvert W \right \rvert \right]^2 = \amen_W \Var_W.\]
\end{proof}

We now consider the general problem of one-bit quantization of random variables in Hilbert space. We first show that the variance drop of a random variable in Hilbert space is the supremum of the variance drop of its projection over all directions. If the projection is symmetric and log-concave for every direction then using Lemma~\ref{lem:vardrop}, the variance drop of the projection can be related to its amenability. 

\begin{theorem}\label{thm:amen}
	Let $\hil$ be a Hilbert space with a countable basis and let $X$ be a zero-mean, finite variance random variable in $\hil$. The following are true. 
	\begin{enumerate}[label=(\alph*)]
		\item $\vdrop_{X} = \sup\limits_{q \in \hil, \left \lVert q \right \rVert = 1} \vdrop_{\left \langle X,q \right \rangle}$
		\item If $\left \langle X ,q \right \rangle$ is symmetric and log-concave for all $q$, then 
		\[ \vdrop_X = \sup\limits_{q \in \hil, \left \lVert q \right \rVert = 1} \amen_{\left \langle X,q\right \rangle} \Var_{\left \langle X,q\right \rangle}.\]
	\end{enumerate}
\end{theorem}

\begin{proof}

(a) Let $Q$ be any one-bit quantizer. Define $q \eqdef \frac{\hat{x}_1 - \hat{x}_0}{\left \lVert \hat{x}_1 - \hat{x}_0 \right \rVert}$. Let $\left \lbrace q, b_1, b_2 \cdots  \right \rbrace$ be an orthonormal basis for $\hil$. Then 
\begin{align}
	 &\EE\left[ \left \lVert X \right \rVert^2\right] - \EE\left[ \left \lVert X - Q(X) \right \rVert^2\right] \nonumber \\
	 &= \EE\left[ \left \langle X,q\right \rangle^2 \right] + \sum\limits_{i=1}^{\infty} \EE\left[ \left \langle X,b_i\right \rangle^2 \right] \nonumber\\
	 &\hspace{5mm}- \EE\left[ \left \lVert \left \langle X,q\right \rangle q + \sum\limits_{i=1}^{\infty} \left \langle X, b_i \right \rangle b_i \right. \right. \nonumber\\
	 & \hspace{10mm}\left. \left. - \left \langle Q(X),q\right \rangle q - \sum\limits_{i=1}^{\infty} \left \langle Q(X),b_i \right \rangle b_i  \right \rVert^2 \right]. \label{eq:decompose}
 \end{align} 

Let $\overline{q} = \Pr\left(f(X)=0\right)x_0 + \Pr\left(f(X)=1\right)x_1$. Then
\begin{align*}
	\left \langle Q(X),b_i \right \rangle &= \left \langle Q(X) + \overline{q} -\overline{q}, b_i \right \rangle \\
	&= \left \langle Q(X)-\overline{q}, b_i \right \rangle + \left \langle \overline{q},b_i \right \rangle = \left \langle \overline{q},b_i \right \rangle,
\end{align*}
 where the last equality is since $Q(X)-\overline{q}=cq$ for $c \in \RR$ and is orthogonal to $b_i$. Substituting in \eqref{eq:decompose}, 
 \begin{align}
 	&\EE\left[ \left \lVert X \right \rVert^2\right] - \EE\left[ \left \lVert X - Q(X) \right \rVert^2\right] \nonumber \\
 	&= \EE\left[ \left \langle X,q\right \rangle^2 \right] - \EE\left[ \left( \left \langle X,q\right \rangle - \left \langle Q(X),q\right \rangle \right)^2 \right] \nonumber \\
 	&+ \sum\limits_{i=1}^{\infty} \EE\left[ \left \langle X,b_i\right \rangle^2 \right] - \sum\limits_{i=1}^{\infty} \EE\left[ \left( \left \langle X,b_i \right \rangle - \left \langle \overline{q},b_i \right \rangle \right)^2 \right].  \nonumber\\
 	&\leq \EE\left[ \left \langle X,q\right \rangle^2 \right] - \EE\left[ \left( \left \langle X,q\right \rangle - \left \langle Q(X),q\right \rangle \right)^2 \right] \nonumber\\
 	&\leq \vdrop_{\left \langle X,q\right \rangle}. 
 \end{align} 
 Since $Q(\cdot)$ was arbitrary, 
 \[ \vdrop_{X} \leq \sup\limits_{q \in \hil, \left \lVert q \right \rVert =1} \vdrop_{\left \langle X,q \right \rangle}.\]
 
Conversely, take any $q\in \hil$ such that $\left \lVert q \right \rVert =1$. Let $Q(\cdot)$ be a one-bit quantizer on $\RR$ satisfying
\begin{align}
	\EE\left[ \left \langle X,q \right \rangle^2 \right] &- \EE\left[ \left( \left \langle X,q \right \rangle - Q\left( \left \langle X,q\right \rangle \right)\right)^2 \right] \nonumber \\
	&\geq \vdrop_{\left \langle X,q \right \rangle} - \eps. \label{eq:approx}
\end{align}
Construct a one-bit quantizer $Q^*(\cdot)$ on $\hil$ where 
\begin{align*}
	g^*(0) = g(0)q, g^*(1)=g(1)q, 
\end{align*}
and $f^*(x) = f\left( \left \langle x,q\right \rangle \right)$. Then
\begin{align*}
	\vdrop_X \geq \EE\left[ \left \lVert X \right \rVert^2\right] - \EE\left[ \left \lVert X - Q^*(X) \right \rVert^2\right]. 
\end{align*}
Let $\left \lbrace q, b_1, b_2 \cdots \right \rbrace$ be an orthonormal basis in $\hil$. Note that $\left \langle Q^*(x), b_i \right \rangle =0$ for all $i$ and $x$. Using the decomposition in \eqref{eq:decompose}, we have
\begin{align*}
	\vdrop_X &\geq \EE\left[ \left \langle X,q\right \rangle^2 \right] + \sum\limits_{i=1}^{\infty} \EE\left[ \left \langle X,b_i\right \rangle^2 \right] \nonumber\\
	&\hspace{5mm}- \EE\left[ \left \lVert \left \langle X,q\right \rangle q + \sum\limits_{i=1}^{\infty} \left \langle X, b_i \right \rangle b_i \right. \right. \\
	&\hspace{10mm}\left. \left. - \left \langle Q^*(X),q\right \rangle q  \right \rVert^2 \right] \\
	&=\EE\left[ \left \langle X,q\right \rangle^2 \right] - \EE\left[ \left( \left \langle X,q\right \rangle - \left \langle Q^*(X),q\right \rangle \right)^2 \right] \\
	&= \EE\left[ \left \langle X,q\right \rangle^2 \right] - \EE\left[ \left( \left \langle X,q\right \rangle - Q\left(\left \langle X,q\right \rangle\right) \right)^2 \right] \\
	&\geq \vdrop_{\left \langle X,q \right \rangle} - \eps.
\end{align*}
But $\eps$ and $q$ were arbitrary. Therefore, 
\[ \vdrop_X \geq \sup\limits_{q \in \hil, \left \lVert q \right \rVert = 1} \vdrop_{\left \langle X,q \right \rangle}.\]

(b) From \cite{Kieffer83}, we know that the unique optimal one-bit quantizer of a symmetric real-valued random variable with log-concave pdf is symmetric. Therefore, the result follows from (a) and Lemma~\ref{lem:vardrop}.
\end{proof}
Since the optimal direction to project along requires that the product of amenability and variance of the projection be maximum, projecting along the direction of highest variance need not always be optimal. We now look at an example that illustrates this point. 

\noindent \textbf{Example:} Let $\overline{S} = \begin{bmatrix}
	S_1,
	S_2
\end{bmatrix}$ where $S_1$ and $S_2$ are independent Laplace random variables with mean zero and variance 2. We will show that projecting along $\begin{bmatrix}
\frac{1}{\sqrt{2}},
\frac{1}{\sqrt{2}}
\end{bmatrix}$ results in a higher variance drop compared to projecting along either of the coordinate vectors. First note that since $\overline{S}$ is a symmetric, log-concave random vector, Theorem~\ref{thm:amen} holds. Therefore, it is sufficient to prove that $\EE\left[ \left \lvert \frac{S_1 + S_2}{\sqrt{2}} \right \rvert \right] > \EE\left[ \left \lvert S_1\right \rvert \right] =  \EE\left[ \left \lvert S_2\right \rvert \right]$. The pdf of $S_1 + S_2$ is $\frac{1}{4}e^{-\left \lvert z \right \rvert}\left(|z|+1\right)$. Therefore, 

\begin{align*}
	\EE\left[ \left \lvert \frac{S_1 + S_2}{\sqrt{2}} \right \rvert \right] &= \frac{1}{\sqrt{2}} \left[ 2 \int_0^{\infty} z \cdot \frac{(z+1)e^{-z}}{4} dz \right] \\
	&= \frac{3}{2\sqrt{2}} > \EE\left[ \left \lvert S_1\right \rvert \right] = \EE\left[ \left \lvert S_2\right \rvert \right] = 1.
\end{align*}

%% file: stsaw.tex
We now consider an application of the previous setup to find  the optimal one-bit quantizer of the stationary sawbridge. Wagner and Ball{\'e} \cite{WagnerB21} studied the sawbridge process, which is defined as 
\[ \saw_t \eqdef t - \bm{1}\left( t \geq \drop\right)\hspace{5mm} t \in \left[0,1\right], \] 
where $\drop\sim \text{Unif}\left[0,1\right]$. We denote the entire process $\left \lbrace \saw_t \right \rbrace_{t=0}^1$ by $\saw$ and call it the \emph{nonstationary sawbridge} to distinguish it from the \emph{stationary sawbridge} 
\begin{equation} \label{eq:stsawdef}
\stsaw_t \eqdef \saw_{\left(  t+ \phase \right) \text{mod } 1} \hspace{5mm} t \in \left[0,1\right],
\end{equation}
where $\drop, \phase \sim \text{Unif}\left[0,1\right]$ and $\drop \indep \phase$. We denote the entire process $\left \lbrace \stsaw_t \right \rbrace_{t=0}^1$ by $\stsaw$.

Since the stationary sawbridge is a rotation of the nonstationary sawbridge in time, both the processes have the same average value or DC, $\int_0^1 \saw_t dt = \int_0^1 \stsaw_t dt = \drop - 0.5$. For the nonstationary sawbridge, it is known from Corollary 2 in \cite{WagnerB21} that an optimal one-bit quantizer is the sign of the DC. From Theorem~\ref{thm:amen} we know that finding an optimal one-bit quantizer is equivalent to finding an optimal direction to project upon and then quantizing the projection.  It should be noted that the constant function equal to 1 is not the highest variance eigenfunction of $\saw$ providing another instance where projecting along a direction different from the highest variance direction is optimal. As we shall see below, this is not the case for stationary sawbridge. Our main result in this section is that the optimal direction to project upon is the constant function equal to $1$ and therefore, the sign of the DC is an optimal one-bit quantizer for the stationary sawbridge. We now specify the eigenfunctions and eigenvalues of the stationary sawbridge.

\begin{lemma}\label{lem:eig}
	The functions $\eigfunc_{1,t} =1, \eigfunc_{2k,t} = \sqrt{2}\sin \left( 2 \pi k t\right), {\eigfunc_{2k+1,t} = \sqrt{2}\cos \left( 2 \pi k t\right)}$ for ${k\geq 1}$ form an orthonormal basis of $\lebtwo$ and are the eigenfunctions of the stationary sawbridge with eigenvalues $\eigval_1=\frac{1}{12}, \eigval_{2k} = \eigval_{2k+1} = \frac{1}{4\pi^2k^2}$.
\end{lemma}

\ifdefined \EXTENDED The proof of Lemma~\ref{lem:eig} is in section~\ref{subsec:proofs}. \else 
In what follows we omit the proofs of Lemmas~~\ref{lem:eig},\ref{lem:proj}, \ref{lem:smalltheta}, \ref{lem:bigtheta} due to space constraints. Complete proofs can be found in the extended version \cite{BhadaneW22}. \fi

%
%

\begin{theorem} \label{thm:stsaw}
	Let $f^*: \lebtwo \mapsto \left \lbrace 0,1 \right \rbrace$ be defined as $f^*(\stsaw)=1$ if $\int_0^1 \stsaw_t dt > 0$ and $f^*(\stsaw)=0$ otherwise.  Define $g^* : \left \lbrace 0,1\right \rbrace \mapsto \lebtwo$ as $g^*(0) = -0.25, g^*(1) = 0.25$.  Then $g^* \circ f^*$ is an optimal one-bit quantizer of $\stsaw$. 
\end{theorem}

 \begin{proof}

 From Theorem~\ref{thm:amen}, we know that \[\vdrop_{\stsaw} = \sup\limits_{q \in \lebtwo, \left \lVert q \right \rVert = 1}\vdrop_{\int_{0}^{1}  q_t \stsaw_t dt   }  \]
 
 Therefore, finding the unit norm function $q$ that maximizes the variance drop of the projection is sufficient to obtain an optimal one-bit quantizer of $\stsaw$. Define the projection of $\stsaw_t$ on $q_t$ as $\proj \eqdef \int_0^1 q_t \stsaw_t dt$. Then for $\thresh \in \RR$, an optimal decision rule for quantizing $\proj$ can be written as  
 \[ \proj \stackrel[1]{0}{\lessgtr} \thresh. \]
We prove that $q^{*}_t=1$ is optimal and that the quantizer for this choice is symmetric, $\thresh^*=0$. 
\ifdefined \EXTENDED The proofs of Lemmas~\ref{lem:proj}, \ref{lem:smalltheta}, \ref{lem:bigtheta} are in section~\ref{subsec:proofs}. \else  \fi

 \begin{lemma}\label{lem:proj}
 	For a unit norm $q$, define $\proj \eqdef \int_0^1 q_t\stsaw_t dt$. Let $\theta \eqdef \left(\int_0^1 q_t dt\right)^2$. Then, 
 	\begin{enumerate}
 		\item  $\proj = \sqrt{\theta} \projDC + \sqrt{1-\theta} \projAC$, where $\projDC \eqdef \text{sgn}\left(\int_0^1 q_t dt \right)\int_0^1 \stsaw_t dt$ and $\projAC \eqdef \int_0^1 g_t \stsaw_t dt$ where $g_t$ is unit norm and $\int_0^1 g_tdt = 0$.
 		\item $\projAC$ and $\projDC$ are independent. 
 	\end{enumerate}
 \end{lemma}

	Since $q$ is arbitrary, it suffices to show that $\vdrop_{\projDC} = \max_{\theta \in \left[0,1\right]} \vdrop_{\proj}$. Consider two cases a) $\theta \leq \frac{5}{8}$ and b) $\frac{5}{8} < \theta < 1$. The following lemma proves that the optimal $\theta$ cannot be smaller than $\frac{5}{8}$.

\begin{lemma} \label{lem:smalltheta}
	If $\theta \leq \frac{5}{8}$, $\vdrop_{\proj} \leq \Var_{\proj} < \vdrop_{\projDC}$.
\end{lemma}

For large $\theta$, a variance argument like before does not work because the variance of the DC is high. We use the structure of the probability density function of $\proj$, $f_{\proj}$, to show that the optimal quantizer of $\proj$ is symmetric. 

Let the support of $\sqrt{\theta} \projDC$ be $\left[ -a, a\right]$, and that of $\sqrt{1-\theta} \projAC$ be $\left[ -b,c\right]$ where $c \leq b$ without loss of generality. Note that for $\theta > \frac{5}{8}$, $a>\frac{\sqrt{5}}{4\sqrt{2}}$ and $b < \frac{1}{4\sqrt{2}}$. Also, the support of $\proj$ is $\left[ -\left(a+b\right) , a+c \right]$ with $f_{\proj}(z) = \frac{1}{2a}$ for $z \in \left[ -\left(a-c \right), a-b \right]$. Note that $\frac{a}{2} < a-b$. 

We now construct a random variable $\widetilde{\proj} = \sqrt{\theta} \projDC + \sqrt{1-\theta} \widetilde{\projAC}$, where $\sqrt{1-\theta}\widetilde{\projAC} = -b$ with probability $\frac{c}{c+b}$ and $c$ with probability $\frac{b}{c+b}$. We show that $\vdrop_{\projDC} \geq \vdrop_{\widetilde{\proj}} \geq \vdrop_{\proj}$ with equality holding for $\theta=1$. 

\begin{lemma}\label{lem:bigtheta}
	For $\theta > \frac{5}{8}$, $\vdrop_{\projDC} \geq \vdrop_{\widetilde{\proj}} \geq \vdrop_{\proj}$, where equality holds for $\theta=1$.
\end{lemma}

Therefore, the optimal direction to quantize is $q^*_t=1$ and the optimal quantizer of the projection is symmetric because the uniform distribution is log-concave. This corresponds to the encoder $f^*(\stsaw) = 1$ if $\projDC > 0$ and $f^*(\stsaw)=0$ otherwise. By the Lloyd-Max conditions, the reconstructions are given by $g^*(1) = \EE\left[ \stsaw_t \mid f^*(\stsaw)=1 \right] = 0.25$ and $g^*(0) = \EE\left[ \stsaw_t \mid f^*(\stsaw)=0 \right] = -0.25$.

\end{proof}

\ifdefined\EXTENDED
 
 \subsection{Proofs of Lemmas }\label{subsec:proofs}
 We list the proofs of unproven lemmas here. 
 \begin{proof}[Proof of Lemma~\ref{lem:eig}]
 
 Define $\rphase_t \eqdef \left( t + \phase \right) \text{mod} 1$. The autocorrelation of $\stsaw_t$ is
 \begin{align*}
 	\acorr(s,t) &= \EE\left[ \stsaw_s \stsaw_t \right] \\
 	&= \EE\left[ \left( \rphase_s - \bm{1}\left( \rphase_s \geq \drop \right) \right) \left( \rphase_t - \bm{1}\left( \rphase_t \geq \drop \right) \right) \right] \\
 	&= \EE\left[ \rphase_s \rphase_t \right] + \EE\left[ \bm{1}\left( \min\left( \rphase_s, \rphase_t \right) \geq \drop \right) \right] \\
 	&  \hspace{5mm}- \EE\left[ \rphase_s \bm{1} \left( \rphase_t \geq \drop \right) \right] - \EE\left[ \rphase_t \bm{1} \left( \rphase_s \geq \drop \right) \right] \\
 	&= \frac{\left(s-t\right)^2}{2} - \frac{|s-t|}{2} + \frac{1}{6}. 
 \end{align*}
 If $\{ \eigfunc_{k,t} \}_{k=1}^{\infty}$ and $\{ \eigval_k \}_{k=1}^{\infty}$ are the eigenfunctions and eigenvalues of $\acorr$, then for all $k$ and $s\in \left[0,1\right]$,
 \[ \int_0^1 \acorr(s,t) \eigfunc_{k,t} dt =\eigval_k \eigfunc_{k,s}. \]
 By differentiating both sides w.r.t $s$ and solving the resultant differential equation, it can be shown that the eigenfunctions are $\eigfunc_{1,t} =1, \eigfunc_{2k,t} = \sqrt{2}\sin \left( 2 \pi k t\right), \eigfunc_{2k+1,t} = \sqrt{2}\cos \left( 2 \pi k t\right)$ for $k\geq 1$. The corresponding eigenvalues are $\eigval_1=\frac{1}{12}, \eigval_{2k} = \eigval_{2k+1} = \frac{1}{4\pi^2k^2}$.
\end{proof}
 \begin{proof}[Proof  of Lemma~\ref{lem:proj}]
 	Since $q$ is unit norm, $\theta \in \left[0,1\right]$.  We can decompose $q_t$ into its DC and AC,
 	\begin{equation}  \label{eq:dcac}
 		q_t = \text{sgn}\left( \int_0^1 q_t dt \right) \sqrt{\theta} + \sqrt{1-\theta} g_t, 
 	\end{equation}
 	where $g_t$ is unit norm and because of orthogonality, $\int_0^1 g_t dt = 0$. Therefore
 	\[ \proj = \sqrt{\theta} \projDC + \sqrt{1-\theta} \projAC. \]
 	
 	The nonstationary sawbridge can be written as 
 	\[ \saw_t = \left( t - \sawDC - \frac{1}{2} \right)\text{mod} 1 - \frac{1}{2}+ \sawDC \]
 	where $\sawDC \sim \text{Unif}\left[-0.5,0.5\right]$. Thus
 	\begin{align*}
		\stsaw_t &= \left( \left(t + \phase\right)\text{mod } 1 - \sawDC - \frac{1}{2} \right)\text{mod } 1  - \frac{1}{2} + \sawDC \\
		&= \left(t+\phase -\sawDC -\frac{1}{2}\right) \text{mod } 1 - \frac{1}{2} + \sawDC\\
		&= \left( t+ \left(\phase - \sawDC \right) \text{mod } 1 - \frac{1}{2} \right) \text{mod } 1- \frac{1}{2} + \sawDC.
 	\end{align*}
 Since $\left(\phase - \sawDC\right) \text{mod } 1$ is independent of $\sawDC$ and since $\projDC$ depends only on $\sawDC$ and $\projAC$ depends only on  $\left(\phase - \sawDC\right) \text{mod } 1$, $\projDC$ and $\projAC$ are independent. 
 	
%
 	
 \end{proof}

\begin{proof}[Proof of Lemma~\ref{lem:smalltheta}]
	$\var{\projDC}=\frac{1}{12}$. By the Karhunen-Lo\`{e}ve theorem, we can express $\stsaw$ as 
	 	\begin{equation}\label{eq:klstsaw}
		 		\stsaw_t = \coeff_1 \eigfunc_{1,t} + \sum\limits_{k=2}^{\infty} \coeff_k \eigfunc_{k,t},
		 	\end{equation}
	 	where $\{ \eigfunc_{k,t} \}_{k=1}^{\infty}$ are eigenfunctions of $\acorr$, and $\coeff_k \eqdef \int_0^1 \stsaw_t \eigfunc_{k,t} dt$ for $k \geq 1$. By Lemma \ref{lem:eig}, since $\{ \eigfunc_{k,t} \}_{k=1}^{\infty}$ is an orthonormal basis for $\lebtwo$, for $\{ \gcoeff_k \}_{k=1}^{\infty} \in \RR$, we can represent $g$ as 
	 	\begin{equation}
		 		g_t = \gcoeff_1 \eigfunc_{1,t} + \sum\limits_{k=2}^{\infty} \gcoeff_k \eigfunc_{k,t}. 
		 	\end{equation}
	 	Since $\int_0^1 g_t dt = 0$ and $\eigfunc_{k,t}$ is orthogonal to $\eigfunc_{1,t}=1$ for $k\geq 2$, $\gcoeff_1=0$. This implies 
	 	
	 	\begin{equation}
		 		\projAC = \int_0^1 g_t \stsaw_t dt = \sum\limits_{k=2}^{\infty} \coeff_k \gcoeff_k. 
		 	\end{equation}
	Since $\projAC = \sum\limits_{k=2}^{\infty} \coeff_k \gcoeff_k$, 
	\[ \var{\projAC} = \sum\limits_{k=2}^{\infty} \var{\coeff_k} \gcoeff_k^2 = \sum\limits_{k=2}^{\infty} \eigval_k \gcoeff_k^2. \] 
	Further, since $g$ is unit norm, $\sum\limits_{k=2}^{\infty} \gcoeff^2_k = 1$. Therefore, 
	\[ \var{\projAC} \leq  \max_{k \geq 2} \eigval_k  = \frac{1}{4\pi^2}. \] 
	For $\sawDC = \drop - 0.5$, 
	\begin{align*}
		\projAC &= \int_0^1 g_t \stsaw_t dt \\
		&= \int_0^1 g_t \left(\stsaw_t - \sawDC + \sawDC \right) dt \\
		&= \int_0^1 g_t \left( \stsaw_t - \sawDC \right) dt \\
		&\leq \sqrt{\int_0^1 \left( \stsaw_t -\sawDC \right)^2 dt } \leq \frac{1}{\sqrt{12}}. 
	\end{align*}
	Therefore, $\projAC$ lies within $\left[ -\sqrt{\frac{1}{12}}, \sqrt{\frac{1}{12}}\right]$ almost surely.
	
	For $\theta \leq  \frac{5}{8}$,
	\begin{align*}
		\var{\proj} &= \theta \var{\projDC} + \left(1-\theta\right) \var{\projAC} \\
		&\leq \frac{1}{4\pi^2} + \frac{5}{8}\left( \frac{1}{12} - \frac{1}{4\pi^2} \right) \leq \frac{1}{16}.
	\end{align*}
\end{proof}

%
%

\begin{proof}[Proof of Lemma~\ref{lem:bigtheta}]

We first prove that for both $\proj$ and $\modproj$ the median is $0$. 

\begin{align}
	\Pr\left( \proj \geq 0 \right) &= \int_0^{a+c} f_{\proj}(z)dz \nonumber \\
	&= \int_0^{a-b} \frac{1}{2a} dz + \int_{a-b}^{a+c} f_{\proj}(z)dz \nonumber\\
	&= \frac{1}{2} - \frac{b}{2a} + \int_{a-b}^{a+c} f_{\proj}(z)dz, \label{eq:median_1}
\end{align} 
 where $f_{\proj}$ is the pdf of $\proj$. Since $\proj$ is the sum of independent random variables, $f_{\proj}$ can be written as a convolution of $f_{ \sqrt{\theta} \projDC}$ and $f_{\sqrt{1-\theta}\projAC}$. For simplicity of notation we denote the pdf of $\sqrt{1-\theta}\projAC$ as $f_{AC}$ and denote its cumulative distribution function (cdf) as $F_{AC}$.  
 \begin{align}
 	\int_{a-b}^{a+c} f_{\proj}(z)dz &= \int_{a-b}^{a+c} \left(\int_{a}^{z-c} \frac{1}{2a} f_{AC} \left(z - \tau \right) d\tau \right) dz \nonumber \\
 	&=\frac{1}{2a}\int_{a-b}^{a+c} \left(\int_{a}^{z-c} f_{AC} \left(\gamma \right) d\gamma \right) dz \nonumber \\
 	&= \frac{1}{2a} \int_{a-b}^{a+c} 1 - F_{AC}\left(z-a\right) dz \nonumber \\
 	&= \frac{b}{2a}, \label{eq:median_2}
 \end{align}
where in the last equality we use the identity $\int_{\ell}^{u} F(x) dx = u - \EE\left[X\right]$ for a random variable $X$ with cdf $F$ whose support is $\left[\ell,u\right]$ where $\ell, u \in \RR$. Substituting \eqref{eq:median_1} in \eqref{eq:median_2}, we get $\Pr\left(\proj \geq 0 \right) = \Pr\left( \proj<0\right) = \frac{1}{2}$. Note that for the proof above we only require that $\sqrt{1-\theta} \projAC$ is supported on $\left[ -b,c\right]$ and its mean is $0$. Therefore, $\Pr\left(\modproj \geq 0 \right) = \Pr\left( \modproj<0\right) = \frac{1}{2}$. 

We now compute $\vdrop_{\modproj}$ using Lemma~\ref{lem:vardrop}. The optimal $\vthresh$ in \eqref{eq:vardrop_nonopti} lies in $\left[ -\frac{a}{2}, \frac{a}{2}\right]$. For $\vthresh \in \left[-\frac{a}{2}, \frac{a}{2}\right]$, $\Pr\left( \modproj \geq \vthresh \right) = \frac{1}{2} -\frac{\vthresh}{2a}$ and 
\begin{align*}
	\EE\left[ \modproj \mid \modproj \geq \vthresh \right] &= \frac{1}{\frac{1}{2} - \frac{\vthresh}{2a}}\left[ \int_{\vthresh}^{a-b} \frac{z}{2a} dz \right. \\
	& \left. \hspace{5mm}+ \int_{a-b}^{a+c} z \frac{1}{2a}\frac{b}{\left(c+b\right)} \right] \\
	&= \frac{1}{4a\left(\frac{1}{2}-\frac{\vthresh}{2a} \right) } \left( a^2 + bc -\vthresh^2 \right). 
\end{align*}
 
 It can be shown that 
 \begin{align*}
 	 &\argmax\limits_{\vthresh \in \left[ -\frac{a}{2}, \frac{a}{2} \right]} \EE\left[ \modproj \mid \modproj \geq \vthresh \right]^2 \frac{\Pr\left( \modproj \geq \vthresh\right)}{\Pr\left(\modproj<\vthresh\right)} \\
 	&= \argmax\limits_{\vthresh \in \left[ -\frac{a}{2}, \frac{a}{2} \right]} \frac{\left( \frac{a^2 + bc - \vthresh^2}{2a} \right)^2}{1 - \frac{\vthresh^2}{a^2}}  =0.\\
 \end{align*}
Therefore, 
\begin{equation}\label{eq:vdrop_modproj}
	\vdrop_{\modproj} = \left( \frac{a^2  +bc}{2a} \right)^2 \leq \vdrop_{\projDC} =\frac{1}{16}
\end{equation}
for $a = \frac{\sqrt{\theta}}{2}$ and $b = \frac{\sqrt{1-\theta}}{\sqrt{12}}$. Equality holds for $\theta=1$. 

We now show that $\vdrop_{\proj} \leq \vdrop_{\modproj}$. We again note that for the optimal quantizer of $\proj$, $\vthresh \in \left[ -\frac{a}{2}, \frac{a}{2} \right]$. Therefore, since $\Pr\left(\proj\geq 0\right) = \Pr\left(\proj < 0 \right) = \frac{1}{2}$, 

\begin{equation}\label{eq:vdrop_proj}
\EE\left[ \proj \mid \proj \geq \vthresh \right]^2 \frac{\Pr\left( \proj \geq \vthresh\right)}{\Pr\left(\proj<\vthresh\right)} = \frac{\left(\int_{\vthresh}^{a+c}  z f_{\proj}\left(z\right) dz \right)^2}{\frac{1}{4}- \frac{\vthresh^2}{4a^2}}. 
\end{equation}

\begin{equation}\label{eq:vdropproj_1}
	\int_{m}^{a+c} z f_{\proj}\left(z\right) dz = \int_{\vthresh}^{a-b} \frac{z}{2a} dz + \int_{a-b}^{a+c} z f_{\proj}\left(z\right) dz. 
\end{equation}
Note that from \eqref{eq:median_2}, 
\begin{align}
	\int_{a-b}^{a+c} z f_{\proj}\left(z\right) dz &= \int_{a-b}^{a+c} \frac{z}{2a} \left( 1-F_{AC}\left(z-a\right) \right) dz \nonumber \\
	&= \frac{1}{2a}\int_{-b}^{c}  \left(a+\tau \right) \left( 1-F_{AC}\left(\tau\right) \right) d\tau. \label{eq:exp_cdf}
\end{align} 

Integrating by parts, we have 
\begin{align}
    \nonumber
    &	\int_{-b}^{c}  \left(a+\tau \right)  \left( 1-F_{AC}\left(\tau\right) \right) d\tau  \\
    & = ab-\frac{b^2}{2} +\int_{-b}^{c} \left(a\tau + \frac{\tau^2}{2} \right) f_{AC}\left(\tau\right) d\tau \nonumber \\
	&= ab - \frac{b^2}{2} + \frac{\int_{-b}^{c} \tau^2 f_{AC}(\tau) d\tau}{2}. \label{eq:byparts}
\end{align}
 We now prove that the last term is bounded by $\frac{bc}{2}$. 
 Since $\tau^2$ is convex, by Jensen's inequality we have 
 \[ \tau^2 \leq b^2 \left(1-\frac{\tau+b}{b+c}\right) + c^2\left(\frac{\tau+b}{b+c}\right) \]
 Thus we have
 \begin{equation} \label{eq:varub}
 	\frac{\int_{-b}^{c} \tau^2 f_{AC}(\tau) d\tau}{2} \leq \frac{b^2 \frac{c}{b+c} + c^2\frac{b}{b+c}}{2} = \frac{bc}{2}
 \end{equation}
where we use the fact that the mean of $\sqrt{1-\theta}\projAC$ is $0$.  Substituting \eqref{eq:varub} in \eqref{eq:byparts}, 
 
 \begin{equation}\label{eq:part1}
 		\int_{-b}^{c}  \left(a+\tau \right)  \left( 1-F_{AC}\left(\tau\right) \right) \leq ab - \frac{b^2}{2} + \frac{bc}{2}.
 \end{equation}
 
 Substituting \eqref{eq:part1} in \eqref{eq:exp_cdf},
 
 \begin{align*}
 		\int_{a-b}^{a+c} z f_{\proj}\left(z\right) dz &\leq ab - \frac{b^2}{2} + \frac{bc}{2} \\
 		&= \int_{a-b}^{a+c} z \frac{b}{b+c} dz. 
 \end{align*}

 Therefore, 
 \begin{align}
 	\int_{\vthresh}^{a+c} z f_{\proj} \left(z\right) dz &\leq \int_{\vthresh}^{a-b} \frac{z}{2a}dz +  \int_{a-b}^{a+c} \frac{z}{2a} \cdot \frac{b}{b+c} dz \nonumber \\
 	&= \int_{\vthresh}^{a+c} z f_{\modproj}\left(z\right)dz. \label{eq:equivmodproj}
 \end{align} 
Substituting \eqref{eq:equivmodproj} in \eqref{eq:vdrop_proj}, 
\begin{align*}
	&\EE\left[ \proj \mid \proj \geq \vthresh \right]^2 \frac{\Pr\left( \proj \geq \vthresh\right)}{\Pr\left(\proj<\vthresh\right)} \\ 
	& \hspace{5mm}\leq \EE\left[ \modproj \mid \modproj \geq \vthresh \right]^2 \frac{\Pr\left( \modproj \geq \vthresh\right)}{\Pr\left(\modproj<\vthresh\right)}. 
\end{align*}
Therefore, \[ \vdrop_{\proj} \leq \vdrop_{\modproj} \leq \vdrop_{\projDC}.\]
\end{proof}
\else 

\fi 

%% file: expt.tex
We experimentally verify that the optimal one-bit quantizer of the stationary sawbridge is found by neural-network-based variable-rate compressors trained using stochastic gradient descent (SGD).  A neural-network-based compressor consists of an encoder-decoder pair and a factorized entropy model for entropy coding of the latent components. All three components are implemented using fully connected neural networks as in \cite{WagnerB21} and are trained using the nonlinear transform coding approach in \cite{BalleCMSJAHT20}.
A single realization of the stationary sawbridge is a vector of 1024 equally spaced points between $0$ and $1$. At train time, this vector is passed through the encoder and the output of the encoder is quantized using a differentiable approximation of rounding by soft-rounding and adding uniform noise \cite{AgustssonT20}. The soft-quantized latents are then fed to the decoder to obtain the reconstruction. At test time, the latents are quantized by rounding to the nearest integer. The objective function is the rate-distortion Lagrangian where the rate is computed by the entropy model, and the distortion is the mean-squared error between the inputs and the reconstructions. The encoder, decoder and the entropy model are trained using SGD until convergence. 


%

\begin{figure}[hbt]
	\centering 
	\includegraphics[width=0.95\linewidth]{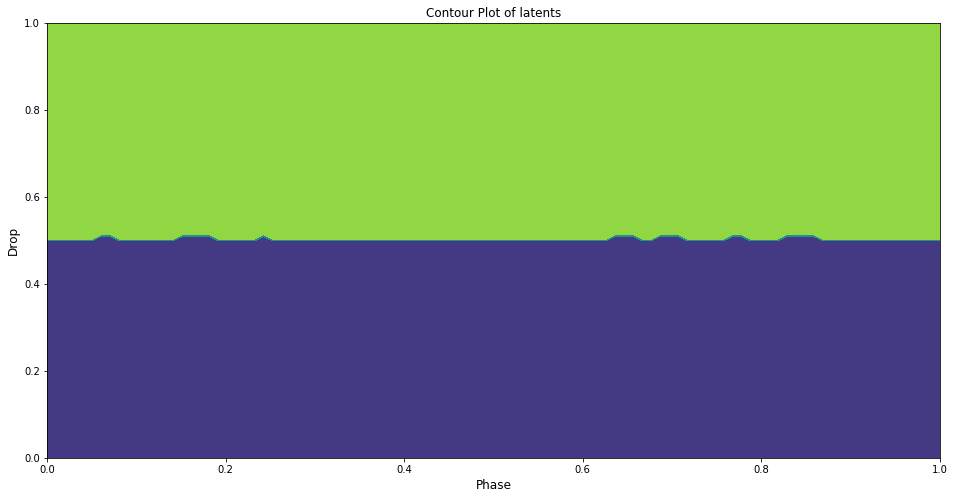}
	\caption{Contour plot for stationary sawbridge.}\label{fig:contour}
\end{figure}

%

 Fig~\ref{fig:contour} is a contour plot of the quantized latent as we vary the drop and phase parameter corresponding to variables $\drop$ and $\phase$ in \eqref{eq:stsawdef}. Note that the quantized latents are the quantized encoder outputs that are then fed to the decoder. Each of the two shaded regions of Fig~\ref{fig:contour} correspond to a single quantized latent vector that differ only in a single latent component. Since the regions correspond to whether the drop is greater than 0.5 or not, neural-network-based compressors trained using SGD converge to an optimal one-bit quantizer. 